\documentclass{article}
\usepackage{fullpage,setspace}

\usepackage{amsmath,amssymb,amsthm, graphicx,subfigure,color}
\usepackage{appendix}
\usepackage{epsfig,enumerate}
\usepackage{amsmath}
\usepackage{amssymb}
\usepackage{color}
\usepackage{url}
\usepackage{constants}

\newcommand{\N}{\mathbb{N}}
\newcommand{\R}{\mathbb{R}}

\newcommand{\A}{\mathcal{A}}
\newcommand{\D}{\mathcal{D}}
\newcommand{\Q}{\mathcal{Q}}

 \newcommand{\sgn}{\operatorname{sign}}

\newtheorem{thm}{Theorem}[]

\newtheorem{prop}[thm]{Proposition}
\newtheorem{cor}[thm]{Corollary}
\theoremstyle{remark}
\newtheorem{rem}{Remark}[thm]
\theoremstyle{definition}
\newtheorem{defn}[thm]{Definition}

\newcommand{\sd}{\Sigma\Delta}
\newcommand{\supp}{\text{supp}}
\newcommand{\gt}{\tilde{\gamma}}

\theoremstyle{definition}
\newtheorem*{rem*}{Remark}

\providecommand{\fr}{\frac}
\newlength{\figurewidth}
\newlength{\smallfigurewidth}

\newcounter{cte}

\title{Quantization of compressive samples with stable and robust recovery}
\author{%
Rayan Saab$^{\ast}$, Rongrong Wang$^{\dag}$, and {\"O}zg{\"u}r Y{\i}lmaz$^{\dag}$\\[0.5em]
{\small\begin{minipage}{\linewidth}\begin{center}
\begin{tabular}{ccc}
$^{\ast}$The University of California, San Diego & \hspace*{0.5in} & $^{\dag}$The University of British Columbia \\
9500 Gilman Dr
 && 121-1984 Mathematics Rd \\
La Jolla, CA 92093, USA && Vancouver, BC V6T 1Z2, Canada\\
\url{rsaab@ucsd.edu} && \url{rongwang@math.ubc.ca}\\ &&\url{oyilmaz@math.ubc.ca}
\end{tabular}
\end{center}\end{minipage}}
}
\begin{document}
\maketitle

\thispagestyle{empty}
\begin{abstract}
  In this paper we study the quantization stage that is implicit in any
  compressed sensing signal acquisition paradigm.  We propose using
  Sigma-Delta ($\Sigma\Delta$) quantization and a subsequent
  reconstruction scheme based on convex optimization. We prove that
  the reconstruction error due to quantization decays polynomially in
  the number of measurements. Our results apply to arbitrary signals,
  including compressible ones, and account for measurement noise.
  Additionally, they hold for sub-Gaussian (including Gaussian and
  Bernoulli) random compressed sensing measurements, as well as for
  both high bit-depth and coarse quantizers, and they extend to
  $1$-bit quantization. In the noise-free case, when the signal is
  strictly sparse we prove that by optimizing the order of the
  quantization scheme one can obtain root-exponential decay in the
  reconstruction error due to quantization.

\end{abstract}
%
%
\section{Introduction}
\label{sec:intro}
Compressed sensing (CS)  \cite{CRT05,Donoho2006_CS} is a signal acquisition and reconstruction paradigm that is based on exploiting a fundamental empirical observation: various types of signals admit (approximately) sparse representations with respect to certain bases or frames.   
For compressive \emph{acquisition} of a signal $x \in \R^N$,  one collects relatively few, say $m\ll N$, inner products of the signal with well-chosen \emph{measurement vectors} such as appropriate random vectors, say $\phi_i \in \R^N$, $i \in \{1,...,m\}$. These inner products yield $m$ compressive measurements $y_i = \langle \phi_i, x\rangle + \eta_i$ , with $i \in \{1,...,m\}$ and $m\ll N$. Here $\eta_i$ represents measurement noise.  To \emph{reconstruct} a good approximation to $x$ from these compressive measurements, one uses some sophisticated non-linear decoder, typically based on convex optimization or on greedy methods. Perhaps the most popular reconstruction method uses $\ell_1$-norm minimization. Specifically, $x$ is approximated with $\hat{x}$, the solution to  
\begin{equation} \min_{z\in\R^N} \|z\|_1 \text{ subject to } \|\Phi z-y\|_2 \leq \epsilon. \label{eq:BPDN}   
\end{equation} 
Here $\Phi$ is an $m\times N$ matrix with $\phi_i$ as its $i$th row, $y = (y_i)_{i=1}^m$ is the vector of acquired measurements,  and $\epsilon$ is an upper bound on the $\ell_2$-norm of the noise-vector $\eta=(\eta_i)_{i=1}^m$. 

Vectors that are $k$-sparse have at most $k$ non-zero entries, i.e., they belong to the set $\Sigma_k^N:=\{x\in\R^N, |\supp{(x)}| \leq k \}$.
 Compressed sensing results are typically formulated in terms ${\sigma_k}(x)_{\ell_1}:=\min\limits_{v\in\Sigma_k^N}  \|x-v\|_1$, called the \emph{best $k$-term approximation error of $x$ in $\ell_1$}, which is a function that measures how close $x$ is to being $k$-sparse. It has been shown, e.g., \cite{CRT05,Donoho2006_CS}, that for a wide class of random matrices (including those with i.i.d. \emph{standard} Gaussian or $\pm 1$ Bernoulli entries), with high probability on the draw of the matrix and uniformly on $x$, the solution $\hat{x}$ to \eqref{eq:BPDN} satisfies\footnote{Since we consider CS matrices with unit-variance entries as opposed to those with unit expected-norm columns, the right-hand side of \eqref{eq:BPDN_thm} contains a factor of $\frac{1}{\sqrt{m}}$ affecting the noise term.} 
 \begin{equation}\|x-\hat{x}\|_2 \leq \Cl{c1} \left( \frac{\epsilon}{
       \sqrt{m}} + \frac{{\sigma_k}{(x)_\ell}_1}{\sqrt{k}}
   \right), \label{eq:BPDN_thm}\end{equation} when $m \geq \Cl{c2} k
 \log{(N/k)}$. Accordingly, any $k$-sparse $x\in \R^N$ can be exactly
 recovered from $m\ll N$ measurements. Furthermore, the solution
 $\hat{x}$ of \eqref{eq:BPDN} provides an approximation of $x$ that is
 within the level of measurement noise (i.e., robust to noise) and
 almost as good as the best $k$-term approximation of $x$ in $\ell_1$
 (i.e., stable with respect to model mismatch). The above is a concise
 description of some of the foundational results in compressed
 sensing. Note that the compressed sensing literature has grown
 extensively in recent years. For example, it includes work on expanding the
 classes of random matrices from which $\Phi$ can be drawn (e.g.,
 \cite{krahmer2014suprema, rauhut2012restricted,
   pfander2013restricted,bajwa2007toeplitz}), and on devising highly
 efficient reconstruction algorithms $\Delta$ so that $\Delta(\Phi x)
 = x $ whenever $x\in{\Sigma}_k^N$ and so that $\|\Delta(\Phi x + e) -
 x \|_2$ is small otherwise (e.g.,
 \cite{needell2009cosamp,dai2009subspace, blumensath2009iterative,
   donoho2009message, SY2010}).  It also includes work on incorporating
 additional structural assumptions or prior knowledge on the underlying signals to
 improve reconstruction accuracy or decrease the number of
 measurements (e.g., \cite{baraniuk2010model,FMSY2012,SM2014}).

\subsection{The quantization problem in compressed sensing}

\emph{Compressed} sensing was named so mainly because of its potential
to combine sensing (acquisition) with compression (source coding) --
see, e.g., the discussion in Section I of
\cite{Donoho2006_CS}. However, most of the literature, as we briefly reviewed
above, ignored analog-to-digital (A/D) conversion or
quantization. Accordingly, compressed sensing emerged mostly as a dimension
reduction technique rather than a scheme for compression. 


Here, we focus on a compressive data acquisition model consisting of a
measurement stage that we have outlined above, followed by a
quantization stage where the measurements are replaced with finite bit
streams, and a reconstruction stage where the signal of interest is
approximated from the quantized measurements. The set of signals of
interest contains all \emph{bounded $k$-sparse signals} in $\R^N$
where $k\ll N$ and is denoted by $\Sigma_K^{N,\mu}:=\{x: \ x\in
\Sigma_k^N, \
\|x\|_2<\mu\}$. 
We are also interested in \emph{bounded compressible signals}, i.e.,
signals that can be well approximated in $\Sigma_K^{N,\mu}$. We will
denote the union of these sets by $\mathcal{X}.$
Let $x\in \mathcal{X}$ and let $y=\Phi x$ be the compressive samples of $x$, where $\Phi$ denotes an $m\times N$ compressive measurement matrix. We define a \emph{quantization map} $\mathcal{Q}$ via its action $(\mathcal{Q} y)_i = q_i,$ where $q_i\in \mathcal{A}$, and $\A$ is a finite set called the quantization alphabet.  Thus, the elements of $\A$ can be assigned finite-length binary labels, amenable to digital storage, and one example is the ``1-bit" alphabet $\A=\{-1,1\}$. Given the quantized compressive measurements $\Q(\Phi x)$, we want to recover a good approximation $\hat{x}$ of $x$ via some reconstruction map $\Delta_\Q: \A^m \mapsto \R^N$ with $\hat{x}=\Delta_\Q( \Q(\Phi x)) $ that ensures that the reconstruction error  $\|x-\hat{x}\|$ is small. 

\subsubsection{Entropy numbers and optimal reconstruction error}
To measure the performance of a practical scheme consisting of the three stages discussed above  (i.e., measurement, quantization, and reconstruction) we will examine the tradeoff between rate (total number of bits) and distortion (reconstruction error).  Ideally, one seeks schemes that approach the optimal error decay for the class of signals at hand, say $\mathcal{X}=\Sigma_k^{N,1}$. In fact, the optimal error is induced by the covering number  $\mathcal{N}(\mathcal{X}, \| \cdot \|, \epsilon)$, defined as  
the minimal number of balls, centered in $\mathcal{X}$, and of radius $\epsilon$ (measured in the norm $\|\cdot\|$) whose union contains $\mathcal{X}$. If we were to assign each of the centers a binary label, this would be an optimal encoding of $\mathcal{X}$, yielding an approximation error of $\epsilon$ while using $\mathcal{R}^*(\epsilon)=\log_2 \big(\mathcal{N} ( \mathcal{X}, \| \cdot \|, \epsilon)\big) $ bits (see, e.g.,  \cite{lorentz1996constructive}). Inverting this function we obtain $\epsilon^*(\mathcal{R})$, the optimal error as a function of the number of bits. For example, for $\mathcal{X}=\Sigma_k^{N,1}$ a volume argument reveals the bound ${\epsilon^*}(\mathcal{R})\gtrsim \frac{N}{k} 2^{-\mathcal{R}/k}$.
Moreover, this bound is essentially achievable for large enough $\mathcal{R}$, by simply using $\log_2 { {N}\choose{k}}$ bits to encode the location of the non-zero coefficients of $x\in\mathcal{X}$ and using the remaining bits to directly encode the non-zero entries. We note that such an approach is not available to us, as we only have access to the vector of measurements $\Phi x$ -- and not to $x$ itself. Moreover,  $\Phi x$ must be quantized upon acquisition using analog hardware. Such hardware cannot invert the measurement system to recover $x$ (otherwise there would be no need for quantization). Instead, one seeks quantization schemes that act on $\Phi x$, and reconstruction algorithms to obtain an estimate $\hat{x}$ from $\mathcal{Q}(\Phi x)$. Below, we summarize some of the quantization and reconstruction strategies that have been proposed in the context of compressed sensing. 

\subsubsection{Memoryless scalar quantization with standard reconstruction techniques} Perhaps the most intuitive approach to quantization is the so-called \emph{memoryless scalar quantization} (MSQ). Here one replaces each measurement $y_i$ with the element of the quantization alphabet that best approximates it. Formally, a \emph{scalar quantizater} $Q_\A: \R \to \A$ is defined via
\begin{align}\label{eq:scalar}Q_\A(z)\in\min\limits_{v\in\A}|v-z|.\end{align}
Suppose that we take $m$ measurements of signals in $\mathcal{X}$, MSQ maps the vector of measurements to $\A^m$ by acting independently on each entry:
\begin{align}\label{eq:MSQ}
\mathcal{Q}_{\A}^{MSQ}: \quad  \Phi\mathcal{(X)} &\to \A^{m} \\\nonumber
				 y &\mapsto q, \textrm{ where } q_i  = Q_\A(y_i), \quad  i= 1,...,m.
\end{align}
While attractive for its simplicity and widely assumed in the mathematical treatment of signal acquisition methods, MSQ is generally not optimal when one \emph{oversamples}, i.e., one has more measurements than needed for $\Phi$ to be injective \cite{PSY13}. 
This sub-optimality is because MSQ ignores correlations between measurements (as it acts component-wise). Still, at first glance this sub-optimality is well-tolerated in the compressive sampling framework because this framework does not rely on oversampling. In fact, MSQ yields near-optimal rate-distortion guarantees when one is permitted to take a minimal number of CS measurements and to reduce the reconstruction error (i.e., the distortion) by increasing the size of the quantization alphabet \cite{candes2006encoding}. 

On the other hand, such an approach may not be desirable, or even practical, in applications. It does not permit improving the accuracy of the reconstruction once the acquisition hardware, thus the quantization alphabet, is fixed.  This is observed in the approximation error bound \eqref{eq:BPDN_thm}, which in a quantization setting does not guarantee that the error decays as more measurements are taken. For example, consider the scalar quantization setting where every measurement $y_i$ is replaced by the closest element $q_i $ from  \[\A_L^\delta:=\left\{\pm (2j-1)\delta/2, j\in\{1,...,L\}\right\},\] the so-called \emph{$2L$-level midrise alphabet with step size $\delta$}. Here, when $y_i$ is appropriately bounded, we have $|y_i-q_i|\leq \delta/2$ and so one must take $\epsilon = \sqrt{m}\delta/2$ in \eqref{eq:BPDN} and in the error estimate  \eqref{eq:BPDN_thm}. This shows that using simple MSQ coupled with standard reconstruction algorithms only guarantees a reconstruction error that scales with $\delta$, the step-size in the quantization alphabet. It does not guarantee any error decay with $m$. Consequently, in this set-up our only resource for obtaining small reconstruction error (due to quantization) is to use a small quantization step-size $\delta$, thus a large alphabet. Importantly, this typically implies a high implementation complexity in hardware. 

\subsubsection{Modifications to MSQ, and alternative reconstruction techniques} 
To overcome this shortcoming of MSQ, one may devise more sophisticated reconstruction algorithms to improve the error behavior as the number of measurement increases. For example, this idea was explored in \cite{Jacques2010, Jacques2013} where the $\ell_2$ constraint of \eqref{eq:BPDN} is replaced by an $\ell_p$ constraint with $p>2$. By choosing $p$ optimally, \cite{Jacques2010} shows that the reconstruction error decays as $(\log m)^{-1/2}$ as $m$ increases. A different reconstruction algorithm based on message passing was proposed in \cite{kamilov2012message} with an error that is empirically observed to decay like $1/m$. In fact, lower bounds obtained in the frame quantization setting \cite{GVT98} show that such an error rate is optimal for MSQ. The error associated with this quantization technique can \emph{at best} decay linearly in the number of measurements, regardless of the reconstruction algorithm. This is far from optimal, as \emph{the optimal encoding (over all quantization methods) of sparse signals yields an error that decays exponentially in $m$} (as we discussed above). Note that one can essentially match the optimal bound by introducing dither and modifying the scalar quantizer so that non-contiguous intervals map to the same element of $\mathcal{A}$ \cite{B12}. However, this does not overcome the increased hardware complexity associated with a large alphabet, as the hardware must still distinguish between small intervals in order to map them to the appropriate element of $\mathcal{A}$. Additionally, this approach has theoretical guarantees only for a combinatorially complex, i.e., impractical, decoder.

Alternatively, one may consider improving the rate-distortion tradeoff associated with MSQ by efficiently encoding its resulting bit-stream. This entails finding a simple map from $\mathcal{A}^m$ to a smaller codebook $\mathcal{C}$, with $\log{|\mathcal{C}|}\ll m\log{|\mathcal{A}|}$, and a reconstruction algorithm that guarantees sufficient error decay with $m$. However, we do not know of any such encoding/reconstruction schemes with theoretical guarantees on the decay of the worst-case error.

\subsubsection{One-bit MSQ and its variations}
Despite the shortcomings of MSQ, its extreme case of one-bit quantization has been the focus of several contributions. Here, one simply replaces every measurement $y_i$ by $\sgn{(y_i)}$ and as a consequence loses all information about the magnitude of $y$, hence $x$. Nevertheless, the focus in such a set-up is approximating $x$ up to the lost magnitude information. The first results in this direction \cite{boufounos20081, jacques2013robust} derived lower bounds on the achievable reconstruction accuracy and proposed some heuristic algorithms for estimating $x$. Afterwards, \cite{PV13} proposed estimating $\frac{x}{\|x\|_2}$ as the minimizer of a particular convex optimization problem and showed error decay rates of $\big(\frac{k}{m}\big)^{1/5}$. To remedy the loss of magnitude information, \cite{knudson2014one} recently proposed quantizing $y_i$ to $\sgn{(y_i + \tau_i)}$ for a fixed vector of thresholds $\tau$ and presented associated reconstruction methods. The corresponding reconstruction error decay rate associated with approximating $x$, namely $\big(\frac{k}{m}\big)^{1/5}$, was also derived in \cite{knudson2014one}. 

In another recent work, \cite{baraniuk2014exponential} proposed  quantization schemes that, like  the $\sd$ schemes we focus on, 
adaptively update a quantization threshold as the measurements are taken. That is, in the one-bit case they propose quantizing $y_i$ to $\sgn{(y_i + \tau_i)}$ albeit now the thresholds $\tau_i$ are adaptively chosen. These schemes yield
exponential error decay in $m$, but unlike $\sd$ schemes, 
 require very sophisticated computation to assign
$q$. For each update of the thresholds,  an $\ell_1$-norm optimization problem is solved (or an
iterative hard-thresholding algorithm is run). Furthermore, the scheme entails communicating analog quantities (the current
thresholds) to a digital computer (to solve the optimization problem). It also entails halting the measurement
process until the optimization ends.

\subsubsection{Noise-shaping quantization methods for compressed sensing}

Recently, noise-shaping quantizers \cite{CGKSY15} have been
successfully employed in the compressive sensing
framework. Traditionally such quantizers, for example $\sd$ quantizers,
are used for analog-to-digital conversion of bandlimited signals and
they are designed to push (most of) the quantization error outside of
the signal spectrum (see, e.g., \cite{NST96}). Over the last decade it
has been established that noise shaping quantizers, including $\sd$
quantizers as well as a new family of quantizers based on beta
encoders \cite{CG15} cf. \cite{CGKSY15} (not our focus in
this paper) can be used in the more general setting of linear
sampling systems associated with both frames and compressed sensing.

We remark that $\sd$ quantization, like MSQ, is universal, i.e., it does not require or use information about the vectors $\phi_i$ or $x$. Furthermore, it is progressive, so it continues to operate in the same way as new measurements arrive. In the case of frames, $\sd$ schemes ensure that the quantization
error is close to the kernel of an appropriate reconstruction operator
based on noise-shaping alternative dual frames called Sobolev duals
\cite{blum:sdf, lammers:adf}, see also \cite{GLPSY13, KSY13,
  KSW12}. $\sd$ quantization schemes can also be used in the case of
compressed sensing. Work preceding this paper mainly uses a two-stage
reconstruction scheme \cite{GLPSY13}, see also \cite{KSY13, Feng2014},
where in the first stage one estimates the support of the original
sparse signal and in the second stage, one uses an appropriate Sobolev
dual to reconstruct.

Above, we have provided an overview of various quantization and reconstruction
approaches associated with compressed sensing. These approaches differ in their implementation 
complexity, in their reconstruction algorithms, and in the quality of their reconstruction guarantees. We
refer the reader to \cite{BJKS14} for a more comprehensive review of
the literature on quantization in the compressed sensing context. In what follows, we focus on using $\sd$ schemes of order $r$ to quantize
compressive samples of sparse or compressible signals. The $\sd$
schemes we study are low-complexity and do not
place much computational burden on the quantizer. We dedicate
Section \ref{sec:SD_existing} to reviewing the prior work on
$\sd$-quantization in the compressed sensing context.

\subsection{Contributions}

We show that $\sd$ schemes allow accurate reconstruction from
quantized CS measurements via convex optimization. Our work
generalizes in several ways the results of \cite{GLPSY13, KSY13}, which applied to
sparse signals that satisfy a size condition on their smallest
non-zero entries.  In particular, the convex optimization problem we
propose allows us to remove the size condition and to obtain results
that apply for general signals in noisy scenarios. It also allows us to use small, even $1$-bit, quantization  alphabets. If $q$ denotes the
$\sd$ quantization of the noisy compressed sensing measurements $\Phi
x + \eta$ where $\eta$ satisfies $\|\eta\|_\infty \leq \epsilon$, then
we recover $x$ from $q$ by solving
\begin{align}\label{eq:decoder_intro}
(\hat{x},\hat{\nu}) :=  \arg\min\limits_{(z,\nu)}\|z\|_1 \  \text{ subject to } &  \|D^{-r}(\Phi z+\nu-q )\|_2 \leq \gamma(r)\sqrt{m}  \notag \\
\text{\ \ and\ \ } & \|\nu\|_2\leq \epsilon \sqrt m.
\end{align} 
Here, $D$ denotes the $m\times m$ difference matrix with entries $D_{i,i}=1$, $D_{i+1,i}=-1$ and $D_{i,j}=0$ when $j\notin \{i, i-1\}$. Additionally, $r$ denotes the order of the $\sd$ scheme, and $\gamma(r)$ is a constant associated with the $\sd$ scheme, see Section \ref{sec:SD_existing}. Our main result follows.
\begin{thm}\label{thm:main_intro} Let $k,\ell,m,N$ be integers satisfying $m\geq \ell \geq ck\log(N/k)$ and let $\Phi$ be an $m\times N$ 
sub-Gaussian measurement matrix. With high probability on the draw of $\Phi$, for all $x\in\R^N$, the solution $\hat{x}$ of \eqref{eq:decoder_intro} where $q$ is the $\sd$ quantization of $\Phi x + \omega$ (with $\|\omega\|_\infty\leq \epsilon$) satisfies
\begin{equation}\label{eq:l2err_intro}
\|\hat{x}-x\|_2 \leq C\left(\big(\frac{m}{\ell}\big)^{-r+1/2}\delta+\frac{\sigma_k(x)}{\sqrt k}+\sqrt{\frac{m}{\ell}}\epsilon \right), 
\end{equation}
where $c$, $C$ are constants that do not depend on $m, \ell,N$. 
\end{thm}
We note that Theorem \ref{thm:main} is a more precise version of Theorem \ref{thm:main_intro} with all the probabilities and constants made explicit. 
Our main result shows that $\sd$ quantization of compressed sensing measurements has the following desirable features.
\begin{itemize}
\item \emph{\bf Polynomial error decay}: By \eqref{eq:l2err_intro}, the recovery error associated with quantization decays polynomially in the number of measurements. Polynomial error decay outperforms the lower bound on the error decay associated with memory-less scalar quantization, which is only linear. 

\item \emph{\bf Generality of measurement systems:} The result holds for a broad class of sub-Gaussian compressed sensing matrices. 
\item \emph{\bf Stability and robustness:} It holds for arbitrary signals in the presence of measurement noise. 
\item \emph{\bf Coarse (even 1-bit) quantization:} It applies even in the case of $1$-bit quantization, when appropriate $\sd$ schemes, such as those of \cite{G-exp, DGK10},  are used. 
\end{itemize}


In addition to the above features, Theorem \ref{thm:main_intro}, which holds for various choices of $\ell \in [c k \log(N/k),m]$, has several important implications. Some of these implications, listed below, can be fleshed out by a careful optimization of appropriate parameters in \eqref{eq:l2err_intro} under various signal and noise models. 
\begin{itemize}
\item \emph{\bf Root-exponential error decay for $k$-sparse vectors}
  (Corollary \ref{cor:exp2}): By optimizing the order of the $\sd$
  quantization scheme, one can show that the reconstruction error
  decays root-exponentially in the number of measurements in the
  absence of measurement noise. That is, defining
  $\lambda:=\frac{m}{k\log(N/k)}$, $$\|\hat{x}-x\|_2 \lesssim
  e^{-c\sqrt{\lambda}}.$$ This matches the best known bounds (in the
  absence of further encoding\footnote{It was shown in \cite{IS13}
    that by incorporating an extra encoding stage, exponential
    accuracy in the bit-rate can be achieved. Similar ideas, with
    similar results, can be applied in the compressed sensing
    scenario, but we leave this topic for a different paper.}) for the
  finite frames case of $\sd$ quantization \cite{KSW12}. 

\item \emph{\bf Error decay for compressible vectors} (Corollary
  \ref{cor:compr}): We model a compressible vector $x$ as an element
  of the unit ball of a weak $\ell_p$ space $w\ell_p^N$ with $p\in (0,1)$. That is, the coefficients of $x$ satisfy $|x_{(j)}| \le
  j^{-1/p}$, where $x_{(j)}$ is the $j$th largest-in-magnitude entry of
  $x$. The reconstruction error associated with  such vectors, in the noise-free scenario, satisfies
$$\|\hat{x}-x\|_2\lesssim \left(\frac{\log N}{m}\right)^{\frac{(1/p-1/2)(r-1/2)}{r+1/p-1}}.$$
We note that the lower bound on the reconstruction error, $\varepsilon$, associated with quantizing vectors in $\ell_p$ using $R$ bits, with $R\in[\log_2 N, N]$, is (see, e.g., \cite{schutt1984entropy})
$$\varepsilon \gtrsim \left(\frac{\log(N/R+1)}{R}\right)^{1/p-1/2}.$$
Thus, for a fixed $p$ and using a coarse $\sd$ quantizer (for example, as in \cite{G-exp}), increasing $r$ allows us to approach the optimal error rate for compressible signals. 

\end{itemize}


\section{A/D conversion for compressed sensing by $\sd$ quantization}

\subsection{Basics of $\sd$ quantization}\label{basicsd}
$\sd$ quantizers were introduced by Inose and Yosuda
\cite{inose1963unity} in the context of analog-to-digital conversion
for oversampled bandlimited functions. Because $\sd$ quantization
schemes are robust with respect to certain circuit implementation
imperfections, they have seen widespread practical use since their
introduction, see, e.g., \cite{NST96}. Moreover, the mathematical
literature on $\sd$ quantization has significantly matured over the
last 15 years, see, e.g., \cite{daub-dev, G-exp, DGK10,
  yilmaz2002stability}. For example, as we will discuss below, $\sd$
schemes were shown to be more effective than MSQ in utilizing
redundancy when quantizing redundant frame expansions, see, e.g.,
\cite{Benedetto2006sigma, BPA2007, blum:sdf, KSW12, KSY13},
cf. \cite{PSY13}.

Given a (finite or infinite) sequence $y:=(y_i)_{i\in\mathcal{I}}$, the standard first-order $\sd$ quantizer runs the following iteration for $i\in \mathcal{I}$: 
\begin{align}\label{eq:SDrec}
q_i&=Q_\A(y_i+u_{i-1}) \\
u_i&=u_{i-1}+y_i-q_i \notag.
\end{align}
Here $Q_\A$ is the scalar quantizer defined in \eqref{eq:scalar}. More
generally, an {$r^{\rm th}$-order $\Sigma \Delta$ quantization scheme}
with quantization rule $\rho: \mathbb{R}^{r+1} \rightarrow \mathbb{R}$
and scalar quantizer $Q_\A$
computes $q\in \A^\mathcal{I}$ via the recursion 
\begin{equation}
q_i = Q_\A\left(\rho(y_i,u_{i-1},u_{i-2},\dots,u_{i-r})\right),
\label{equ:rthOrdq}
\end{equation}
\begin{equation}
u_i =  y_i - q_i - \sum^{r}_{j=1} {r \choose j} (-1)^j u_{i-j}
\label{equ:rthOrdu}
\end{equation}
for $i \in \mathcal{I}$.  For example, in the compressed sensing and
finite frame scenarios (see below) $y_i = \langle \phi_i, x \rangle$,
$i= 1,...,m$. In other words $y=\Phi x$ and the relationship between
${ x}$, ${ u}$, and ${ q}$ can be concisely written using
matrix-vector notation
as 
\begin{equation}
D^r { u} ~=~ \Phi x - { q}.
\label{equ:LinrSigDelt}
\end{equation} 
Here the matrix $D$ is the $m\times m$ first-order difference matrix with entries given by \begin{equation}
D_{i,j} := \left\{ \begin{array}{ll} 1 & \textrm{if}~i=j\\ -1 & \textrm{if}~i = j+1\\ 0 & \textrm{otherwise}\end{array} \right..
\label{Def:diffMat}
\end{equation}
Due to the role that $u$ plays in \eqref{equ:LinrSigDelt} and in controlling the reconstruction error, we will focus on  \textit{stable $r${th} order schemes}. These are schemes for which \eqref{equ:rthOrdq} and \eqref{equ:rthOrdu} result in \[\|  u \|_{\infty} \leq  \gamma := \gamma(r)\] for all $m \in \mathbb{N}$, and for all ${ y} \in \mathbb{R}^m$ with $\|{ y}\|_\infty \leq 1$. Clearly, one can scale the definition of stability to replace the upper bound on $\|y\|_\infty$ by an arbitrary constant $\mu$. However, we  require that $\gamma:  \mathbb{N} \mapsto \mathbb{R}^+$ be independent of both $m$ and ${ y}$.  Stable $r${th} order $\Sigma \Delta$ schemes with $\gamma(r) = O(r^r)$ exist \cite{G-exp, DGK10} even when $\mathcal{A}$ is a $1$-bit alphabet. 
We will focus on two important families of stable $\sd$ schemes of arbitrary order. 

\subsubsection*{\bf The $r$th-order greedy $\Sigma\Delta$ quantizer} The $r$th-order greedy $\Sigma\Delta$ quantizer is obtained by setting 
$$
\rho(u_{n-1},...,u_{n-r},y_i):=\sum\limits_{j=1}^r (-1)^{j-1} \binom{r}{j} u_{n-j}+y_n
$$
in \eqref{equ:rthOrdq}.
This choice results in a stable $\sd$ quantizer if we use the $2K$-level midrise alphabet $\A_K^\delta$ with $K\geq 2 \lceil \frac{\|y\|_\infty}{\delta}\rceil +2^r+1$. The size of the alphabet associated with the greedy $r$th-order $\sd$ scheme grows exponentially in $r$. On the other hand, the stability constant \[\gamma(r)=\delta/2\] remains fixed as $r$ grows. 

\subsubsection*{\bf The $r$th-order coarse $\sd$ quantizers}
A different tradeoff between the size of the alphabet and the stability constant  arises with the so-called  \emph{$r$th-order coarse $\sd$ quantizers}. These are $r$th-order $\sd$ quantizers with a fixed  alphabet $\A$ (for example $\A=\{\pm1\}$)  and were first constructed in \cite{daub-dev}. Adopting  
the scheme of \cite{DGK10}, we have
\begin{align}\label{eq:SDrec2}
q_n&=Q_{\A_K^\delta}((h*v)_n+y_n) \\
v_i&=(h*v)_n+y_n-q_n \notag,
\end{align}
where $*$ denotes convolution. In \cite{G-exp, DGK10}, it was shown that by choosing $h$ carefully one obtains a stable $\sd$ scheme when $\|y\|_\infty<\mu$ as above. In fact, when $\|h\|_1 \le 2K-2\mu/\delta$ one can rewrite \eqref{eq:SDrec2} in the form of \eqref{equ:rthOrdq} and \eqref{equ:rthOrdu} via a change of variable. One can then show that the corresponding scheme is stable with constant \[\gamma(r)=\Cl{c3}^r  r^r\delta\]   for some constant $\Cr{c3}$ that only depends on the alphabet $\A_K^\delta$ and $\mu$.  




\subsection{$\sd$ quantization for frame expansions}
As we mentioned above, while initially proposed for quantizing
bandlimited functions, recent work on $\sd$ quantization has shown it
to be well suited for quantizing finite frame expansions (e.g.,
\cite{Benedetto2006sigma, blum:sdf,BPA2007,KSW12})\footnote{A finite frame \cite{casazza2012finite,kovacevic2007life} for $\R^k$ is a collection of $m$ ($m\geq k$) vectors that spans $\R^k$. As such it can be identified (up to permutations) with an $m\times k$ matrix, say $E$, whose rows are the frame vectors.}. In the $\sd$ quantization context, consider full-rank matrices $E\in \R^{m\times k}$, $m\ge k$ selected from certain classes, such as harmonic frames \cite{Benedetto2006sigma}, piecewise-smooth frames \cite{blum:sdf}, Sobolev self-dual frames \cite{KSW12}, Gaussian frames \cite{GLPSY13}, sub-Gaussian frames \cite{KSY13}.  Let $x\in \R^k$ and suppose that $q$ is obtained by quantizing $Ex$ using an $r$th-order $\sd$ quantizer. A typical results states that $E$ admits a specific left inverse (equivalently, a dual frame),  called its $r$th-order Sobolev dual and denoted by $F_{\text{Sob},r}(E)$  such that the reconstruction error $\|F_{\text{Sob},r}(E) q-x\|_2$ decays like $O(\lambda^{-r})$. Here $\lambda:=m/k$ is the oversampling rate. In fact, this error decay is typical of the deterministic frames mentioned above. On the other hand, the error for random frames is $O(\lambda^{-\alpha(r-1/2)})$ with high probability that depends on $\alpha\in (0,1)$.  

This error decay is remarkably similar to the error decay in the
bandlimited case, which is also inverse polynomial in the oversampling
rate. In the bandlimited case, this allows improving the
reconstruction accuracy by sampling at a faster rate. In the finite frame setting one can also improve the reconstruction accuracy without changing the quantization hardware by  simply increasing $m$, i.e., collecting more samples. This allows decreasing the reconstruction error beyond the accuracy of the quantizer itself. 

\subsection{$\sd$ quantization for compressed sensing: an overview of existing results}\label{sec:SD_existing}
As stated in the Introduction, \cite{GLPSY13, KSY13, Feng2014} showed
that by using $\sd$ schemes instead of MSQ, finite frame quantization
techniques (e.g.,
\cite{Benedetto2006sigma,blum:sdf,BPA2007,KSW12,KSY13}) can be
leveraged to the compressed sensing
case. 
A key insight is that knowing the support of an underlying sparse
vector reduces a compressed sensing quantization problem to a frame
quantization problem. To be precise, let $x\in \Sigma_k^N$ with
$\text{supp}(x)=T$ and let $\Phi\in \R^{m\times N}$ be a compressed
sensing measurement matrix. Then, we have
$$y=\Phi x \ \ \implies \ \ y=\Phi_T x_T.$$
Accordingly, \cite{GLPSY13}  proposed using an $r$th order $\sd$-scheme, and a sufficiently fine, fixed alphabet, to quantize compressive samples of \emph{strictly sparse signals} whose smallest non-zero entry is bounded away from zero. Moreover, \cite{GLPSY13}  introduced an associated two-stage algorithm for recovering the signals from the quantized compressed sensing measurements. In the first stage of the algorithm, the signal support is estimated. Here one uses a standard robust compressed sensing decoder, e.g., \eqref{eq:BPDN}, to obtain a coarse estimate of $x$, say $\widetilde{x}$. The support estimate $\widetilde{T}$ is then the subset of $\{1,...,N\}$ on which the $s$ largest (in magnitude) entries of $\widetilde{x}$ reside. When $\Phi$ is a Gaussian random matrix, \cite{GLPSY13} showed that $\widetilde{T}=T$ with high probability if the smallest non-zero entry of $x$ exceeds a constant multiple of the quantizer step size $\delta$. Having revealed the support $T$,  linear reconstruction using \emph{frame theoretic methods} recovers the signal via\[\hat{x}=F_{\text{Sob},r}(\Phi_{\tilde{T}}) q .\]  
The main result of \cite{GLPSY13} was that in this setup (where the signals are strictly sparse and bounded away from zero, and the alphabet is sufficiently large) the reconstruction error associated with an $r$th order $\sd$-scheme decays like $(k/m)^{\alpha(r-1/2)}$, with high probability depending on the parameter $\alpha\in(0,1)$.
These results were generalized to the case of compressed sensing matrices with sub-Gaussian entries in \cite{KSY13} and sub-Gaussian columns in \cite{Feng2014}.

\section{$\sd$ quantization for compressed sensing: a novel one-stage reconstruction method}
\subsection{The algorithm} An important shortcoming of the two-stage
algorithm describe in the previous section is its need for an accurate estimate for the support of $x$. This restricts the use of the algorithm to sparse vectors with smallest (in magnitude) entry bounded away from zero. Moreover, it restricts the quantizer step size $\delta$ to be (up to a constant) as small as the lower bound on the non-zero entries. In turn, this precludes the use of coarse quantization with small alphabets.  An additional shortcoming of the two-stage algorithm (with reconstruction using Sobolev duals) is its lack of robustness to additive noise. 

Instead of the two-stage algorithm, we propose estimating $x$ as the solution to a convex optimization problem that takes $q=Q_{\Sigma\Delta}^r(\Phi x +\eta)$, where $\|\eta\|_\infty\le  \epsilon$, as input. It produces an output $\hat{x}$ via
\begin{align}\label{eq:decoder}
(\hat{x},\hat{\nu}) :=  \arg\min\limits_{(z,\nu)}\|z\|_1 \  \text{ subject to } &  \|D^{-r}(\Phi z+\nu-q )\|_2 \leq \gamma(r)\sqrt{m}  \notag \\
\text{\ \ and\ \ } & \|\nu\|_2\leq \epsilon \sqrt m.
\end{align} 
Here $\gamma(r)$ is the stability constant associated with the $\sd$ scheme. As discussed in Section \ref{basicsd}, when using a midrise quantization alphabet with a greedy quantization scheme, $\gamma(r)=1/2$. Furthermore, when using the coarse scheme in \eqref{eq:SDrec2}, $\gamma(r)=\Cr{c3}^r r^r$.  
One may also estimate $x$ as a solution to a related (and perhaps more natural) optimization problem,  
\begin{align}\label{eq:decoder2}
(\hat{x},\hat{\nu}) :=  \arg\min\limits_{(z,\nu)}\|z\|_1 \  \text{ subject to } &  \|D^{-r}(\Phi z+\nu-q )\|_\infty \leq \gamma(r)  \notag \\
\text{\ \ and\ \ } & \|\nu\|_2\leq \epsilon \sqrt m.
\end{align} 
We remark that both optimization problems are designed to find the vector with the smallest $\ell_1$ norm that agrees with the quantized measurements and the boundedness of the non-quantization noise. The first constraint in each of \eqref{eq:decoder} and \eqref{eq:decoder2} is motivated by the simple observation that the stability of the $\sd$ quantization schemes used implies that $\|D^{-r}(\Phi x+\eta-q )\|_\infty = \|u\|_\infty  \leq \gamma(r) $, which in turn implies that $\|u\|_2 \leq \gamma(r)\sqrt{m}$. 
The remainder of the paper is dedicated to proving that the above reconstruction algorithms have none of the shortcomings of the two-stage algorithm. Our proofs will be for the optimization problem \eqref{eq:decoder} but they apply equally to \eqref{eq:decoder2}, since any feasible point for \eqref{eq:decoder2} is also feasible for \eqref{eq:decoder}.


\subsection{Preliminaries}

We will need the following notation and results. 

\subsubsection{Sub-Gaussian random matrices}
Below, we write $x\sim \D$ when the random variable $x$ is drawn according to the distribution $\D$ and $\mathcal{N}(0,\sigma^2)$ denotes the Gaussian distribution with mean zero and variance $\sigma^2$. 
\begin{defn} Suppose $x \sim \D_1$ and $y \sim \D_2$ are random variables that satisfy $P(|x|>t)\le K P(|y|>t)$ for some constant $K$ and for all $t\ge 0$. Then $x$ is said to be $K$-dominated by $y$. 
\end{defn}
\begin{defn} 
A random variable is sub-Gaussian with parameter $c>0$ if it is $e$-dominated by $\mathcal{N}(0,c^2)$ 
\end{defn}
\begin{defn} A matrix $\Phi$ is said to be sub-Gaussian with parameter $c$, mean $\mu$, and variance $\sigma^2$ if the entries $\Phi$ are indepent sub-Gaussian random variables with parameter $c$, mean $\mu$, and variance $\sigma^2$.\end{defn}

Note that Sub-Gaussian random variables can be equivalently defined using their
moments. Moreover, in the case of zero mean random variables, their moment generating functions can be used. See
\cite{ve12-1} for these definitions and proof of equivalence.

\subsubsection{Some key results from the literature}
Next we review some useful results that are instrumental for the error estimates presented in the next section. 
We begin with a proposition from \cite{KSY13} which controls the restricted isometry constants of some matrices of interest. Here, given a matrix $V$ we denote by $V_\ell$ the matrix formed by the first $\ell$ rows of $V$. 
\begin{prop} [\cite{KSY13}] \label{pro: subG}Let $\Phi$ be an $m\times N $ sub-Gaussian matrix with mean zero, unit variance, and parameter $c$, let $V$ be an orthonormal matrix of size $m\times m$. Fix some $\beta >0$, some $k \in \mathbb{Z}^+$ with $k<\frac{m}{c_\beta\log N/k} $ and some $\ell \in \mathbb{Z}^+$ with  $\ell\geq \Cr{c1} k \log N/k$. Then, there is a $\rho_\beta$ such that
\[
\mathbb{P}\left( \sup_{T\subseteq [N], |T|\leq k}  \| \frac{1}{\ell} \Phi _T^T V_\ell V_\ell^T\Phi_T-I\|_2 \geq \delta \right) \leq e^{-\rho_\beta \ell}
\]
where $c_\beta$ $\rho_\beta$ only depend on $\delta$ and $c$.
\end{prop}
The next proposition provides useful bounds on the singular values of $D^{-r}$. 
\begin{prop}[\cite{GLPSY13}] \label{pro:singular} 
The singular values of the $r$-th order difference matrix $D^r$ satisfy
\begin{equation}\label{eq:sinvalue} 
\sigma_j(D^{-r})\geq\frac{1}{(3\pi r)^r} \left(\frac{m}{j}\right)^r.
\end{equation}
\end{prop}
We will need the following two propositions, the first of which follows from Theorem 5 and Proposition 8 in \cite{Foucart13}. 
\begin{prop}[\cite{Foucart13}]\label{pro:foucart} Let $f,g\in \mathbb{C}^N$, and $\Phi \in \mathbb{C}^{m,N}$. Suppose that $\Phi$ has $\delta_{2k}$-RIP with $\delta_{2k}<1/9$, then for any $1\leq p\leq 2$, we have
\[
\|f-g\|_p\leq \Cl{c10}k^{1/p-1/2} \|\Phi(f-g)\|_2+\frac{\Cl{c11}}{k^{1-1/p}}(\|f\|_1-\|g\|_1+2\sigma_k(g)_1),
\]
where $\Cr{c10}$ and $\Cr{c11}$ are constants that only depend on $\delta_{2k}$.
\end{prop}
\begin{prop}\label{pro:singular_sub}\cite{RV09} Let $\Phi$ be an $m\times N $, $m\geq N$, sub-Gaussian matrix with mean zero, unit variance, and parameter $c$, then the largest and smallest singular values of $\Phi$ obey
\[
P(\sigma_1 (\Phi) \geq  \sqrt m+\sqrt{N} +t ) \leq e^{-\frac{dt^2}{2} }, \quad {and } \quad
P\left( \sigma_m \left(\Phi \right) \leq \sqrt m -\sqrt {N}- t \right) \leq  e^{-\frac{dt^2}{2}}.
\]
where $d>0$ is some constant only depending on $c$.
\end{prop}

\subsection{Error estimates}
The following theorem provides an error bound on the estimate from \eqref{eq:decoder}.
\begin{thm}\label{thm:main} Let $\Phi$ be an $m\times N$ 
sub-Gaussian matrix with mean zero and unit variance, and let $k$ and
  $\ell$ be in $\{1,\dots,m\}$. Denote by $Q_{\sd}^r$ a stable
  $r$th-order scheme with alphabet $\A_K^\delta$ and stability
  constant $\gamma(r)$. There exist positive constants $c$ and
  $\rho$ such that whenever $m\ge
  \ell \ge c k \log N/k$ the following holds with probability
  greater than $1-\exp(-c\rho k\log(N/k))$ on the draw of $\Phi:$

Let $x\in \R^N$ such that $\|\Phi x\|_\infty \le \mu < 1$. Suppose $q:=Q_{\sd}^r(\Phi x+\eta)$,  where the additive noise $\eta$ satisfies $\|\eta\|_\infty\le \epsilon$ for some $0\leq	\epsilon<1-\mu$.  Then the solution $\hat{x}$ to \eqref{eq:decoder} satisfies
\begin{equation}\label{eq:l2err}
\|\hat{x}-x\|_2 \leq \Cl{c4} \left(\frac{m}{\ell}\right)^{-r+1/2}\delta+\Cl{c5} \frac{\sigma_k(x)}{\sqrt k}+\Cl{c6}\sqrt{\frac{m}{\ell}}\epsilon , 
\end{equation}
where $\Cr{c4}$, $\Cr{c5}$, and $\Cr{c6}$ are explicit constants given in the proof.
\end{thm}

\begin{proof} Let $(z,v)$ be an ordered pair that is feasible to
  \eqref{eq:decoder}, and let $\gt:=\gamma(r)/\delta$. Define $u:=D^{-r}(\Phi z+v-q)$,
  and 
  $p:=\left(\frac{1}{\gt}u,\frac{\delta}{\epsilon} v\right)$ and note that
\begin{equation}\label{eq:p}
\|u\|_2\leq \gt \delta\sqrt m, \quad \text{and } \quad \|p \|_2 \leq  \delta \sqrt{2m}.
\end{equation}
By definition, $u$, $p$, $q$, $z$, and $v$ have the relation
\[
\Phi z-q=D^r u+v =\left[\gt D^r,\frac{\epsilon}{\delta} I \right] p, 
\]
which upon defining $H:= [\gt D^r,\frac{\epsilon}{\delta} I ]$ becomes
\begin{equation}
\Phi z-q=Hp.\label{eq:arb}
\end{equation}
Let $U\Sigma V^T=H$ be the singular value decomposition of $H$, and denote by $H^{\dagger} := V^T \Sigma ^{-1} U$ the pseudo-inverse of $H$. Now, apply $H^\dagger$ to both sides of \eqref{eq:arb} to obtain
\begin{equation}\label{eq:sudo}
H^\dagger (\Phi z-q) = H^\dagger Hp = V V^Tp.
\end{equation}
Then, \eqref{eq:p} and \eqref{eq:sudo} together imply that for any $z$ that satisfies the constraint in \eqref{eq:decoder}, we have
\begin{equation}\label{eq:bd}
\|H^\dagger(\Phi z-q)\|_2\leq \delta \sqrt {2m}.
\end{equation}
In particular, both $x$ and $\hat{x}$ satisfy the constraint in \eqref{eq:decoder}, so by the triangle inequality
\begin{equation}\label{eq:lb}
\|H^\dagger\Phi (x-\hat x)\|_2\leq \|H^\dagger(\Phi x-q)\|_2+\|H^\dagger(\Phi \hat x-q)\|_2 \leq 2 \delta \sqrt {2m}.
\end{equation}
On the other hand, the singular values of $H^\dagger$ can be expressed in terms of those of $D^r$. Thus, using Proposition \ref{pro:singular} we have
\begin{equation}\sigma_\ell(H^\dagger) = \left(\gt^2 \sigma_{m-\ell}^2(
  D^r)+\left(\frac{\epsilon}{\delta}\right)^2 \right)^{-1/2}\geq
\left(\gt^2\left(\frac{3\pi
      r\ell}{m}\right)^{2r}+\left(\frac{\epsilon}{\delta}\right)^2
\right)^{-1/2}.\label{eq:sig_l}\end{equation} 
Denoting by $\Sigma^{-1}_\ell$ and $\Phi_\ell$  the first $\ell$ rows of $\Sigma^{-1} $ and $\Phi$, and by $U_\ell$ the first $\ell$ columns of $U$, we may now use the above bound on $\sigma_\ell (H^\dagger)$ along with \eqref{eq:lb}, to obtain
\begin{align}\label{eq:2err}
 2 \delta \sqrt {2m}& \geq \|H^\dagger \Phi (x-\hat x)\|_2 = \|V^T \Sigma^{-1} U\Phi (x-\hat x)\|_2\geq  \|\Sigma^{-1} U\Phi (x-\hat x))\|_2 \\
 &\geq  \|\Sigma^{-1}_\ell U_\ell \Phi_\ell (x-\hat x)\|_2 \geq \sigma_\ell(H^\dagger) \| U_\ell \Phi_\ell (x-\hat x)\|_2   \notag.
\end{align}
Setting $\widetilde{\Phi}:=\frac{1}{\sqrt l} U_\ell \Phi_\ell$, we thus have
\begin{equation}
2\delta\sqrt{2m}\geq  \sigma_\ell(H^\dagger) \sqrt \ell  \| \widetilde{\Phi}_\ell (x-\hat x)\|_2.\label{eq:up_bound}\end{equation}
Moreover, by Proposition \ref{pro: subG} we know that for any $\beta < 1/9$ there exist constants $c$ and $\rho$, (that depend on $\beta$) so that with probability exceeding $1-e^{-\rho \ell}$,  $\widetilde{\Phi}$ has the restricted isometry property with constant $\beta$, provided that $\ell\geq c k \log N/k$. So $\widetilde{\Phi}$ satisfies the hypotheses of Proposition \ref{pro:foucart}, which we can now apply with $p=2$ and with $\hat{x}$ and $x$ in place of $f$ and $g$. By first using the fact that $\|\hat{x}\|_1 \le \|x\|_1$ and then \eqref{eq:up_bound} and \eqref{eq:sig_l} we have
\begin{align*}
  \|x-\hat{x}\|_2 &\leq \Cr{c10} \|\widetilde{\Phi}(x-\hat x)\|_2+\Cr{c11}\frac{\sigma_k(x)_1}{\sqrt k} \\
&\leq 2 \sqrt 2 \Cr{c10} \delta\sqrt{\frac{m}{\ell}}  \frac{1}{\sigma_\ell(H^\dagger)} +\Cr{c11}\frac{\sigma_k(x)_1}{\sqrt k} \\
  & \leq 2\sqrt 2 \Cr{c10} \delta \sqrt
  {\frac{m}{\ell}}\left(\gt^2\left(\frac{3\pi
      r\ell}{m}\right)^{2r}+\left(\frac{\epsilon}{\delta}\right)^2
  \right)^{1/2}+\Cr{c11}\frac{\sigma_k(x)_1}{\sqrt k} \\ & \leq 2\sqrt 2
  \Cr{c10} \delta \sqrt
  {\frac{m}{\ell}}\left(\gt \left(\frac{3\pi r\ell}{m}\right)^{r}+\frac{\epsilon}{\delta}\right)+\Cr{c11}\frac{\sigma_k(x)_1}{\sqrt
    k} \\& \leq 2\sqrt 2 \Cr{c10} \gt 3^r\pi^r r^r
  \left(\frac{\ell}{m}\right)^{r-1/2}\delta +2\sqrt 2 \Cr{c10} \gt^2\sqrt
  {\frac{m}{\ell}}\epsilon+\Cr{c11}\frac{\sigma_k(x)_1}{\sqrt k}
\end{align*}
Setting $\Cr{c4}:= 2^{3/2}\Cr{c10}3^r \pi^r r^r\gt(r) $, $\Cr{c5}:=2\sqrt 2 \Cr{c10} \gt(r)$,  and $\Cr{c6}:=\Cr{c11}$, the error bound \eqref{eq:l2err} is established.
\end{proof}

\begin{rem} The bound \eqref{eq:l2err} shows that the reconstruction error due to  the perturbation $\eta$ scales linearly with $\|\eta\|_2$ (e.g. by setting $\ell=m$). Such perturbations (commonly referred to as measurement noise)  model a variety of phenomena, including circuit imperfections.  On the other hand, $\sd$ quantization is robust to certain circuit imperfections (see, e.g., \cite{daub-dev}, \cite{daubechies2006d}) of which we now provide an example. Suppose that we use a $1$-bit alphabet $\A = \{-\delta/2, \delta/2\}$ and an associated scalar quantizer $Q_\A$ to perform $1$st order $\sd$ quantization, and suppose further that $Q_\A$ is ``imperfect''. In particular, suppose that ${Q}_\A(z) = \sgn(z)$ whenever $|z|>\rho$ for some $\rho \in [0,1)$ but that $Q_\A(z)$ can be $\pm 1$ otherwise, for example at random. In this case, the stability constant (see Section \ref{basicsd}) associated with the scheme simply becomes $\delta/2 + \rho$ (see \cite{daub-dev, daubechies2006d}). As a consequence of our proof, $\delta+2\rho$ replaces $\delta$ in \eqref{eq:l2err} and the contribution of $\rho$ can be made arbitrarily small by taking more measurements. 

\end{rem} 

Below, we prove some important implications of Theorem \ref{thm:main}. Corollary \ref{cor:noise} distinguishes different regimes for the reconstruction error, depending on the relationship between $m$, $\epsilon$, and $\delta$. Corollary \ref{cor:exp2} shows that  choosing the optimal order of a coarse $\sd$ quantization scheme yields root-exponential error decay in the number of measurements.  Corollary \ref{cor:compr}, also dealing with the noise-free scenario, derives the error decay associated with $\sd$ quantization of compressible signals (as a function of $m$).


\begin{cor}\label{cor:noise}
 Let $\Phi$ be an $m\times N$ sub-Gaussian matrix with mean zero and unit variance, and let $k \in \{1,\dots,m\}$ with $m\geq 2c k \log{N}$. Denote by $Q_{\sd}^r$ a stable
  $r$th-order scheme with alphabet $\A_K^\delta$.  
  There exist positive constants $c$ and $\rho$ (as defined in Theorem~\ref{thm:main}) such that the following hold with probability exceeding $1-e^{-c\rho k\log N/k}$ on the draw of $\Phi$:
  
\bigskip
Let $x\in \R^N$ such that $\|\Phi x\|_\infty \le \mu < 1$. Suppose $q:=Q_{\sd}^r(\Phi x+\eta)$,  where the additive noise $\eta$ satisfies $\|\eta\|_\infty\le \epsilon$ for some $0\leq	\epsilon<1-\mu$.  Define $\ell_c:={m}{\left(\frac{\Cr{c6} \epsilon}{\Cr{c4}(2r-1)\delta}\right)^{1/r}}$ and
let $\hat{x}$ be the minimizer of \eqref{eq:decoder}. We distinguish three regimes based on $\ell_c$ and $m$.
\begin{itemize}
\item[]\noindent (i) If $\ell_c \leq  ck\log N$ (the low-noise scenario), then
\begin{equation}\label{eq:noise_a}
\|\hat{x}-x\|_2\leq \Cr{c4} \left(\frac{m}{2c k\log N/k}\right)^{-r+1/2} \delta+\Cr{c5}\fr{\sigma_k(x)}{\sqrt{k}}+\Cr{c6}\sqrt{ \frac{m}{c k\log N/k}} \epsilon.
\end{equation}
\item[]\noindent (ii) If \
 $c k\log N < \ell_c < m$ (the intermediate-noise scenario) then 
\begin{equation}\label{eq:noise}
\|\hat{x}-x\|_2\leq \Cl{c16} \delta^{\frac{1}{2r}} \epsilon^{1-\frac{1}{2r}}+\Cr{c5}\fr{\sigma_k(x)}{\sqrt{k}}.
\end{equation}
\item[]\noindent (iii) If $\ell_c \geq m$ (the high-noise scenario), then 
\begin{equation}\label{eq:noise_b}
\|\hat{x}-x\|_2\leq \Cr{c4} \delta+\Cr{c5}\fr{\sigma_k(x)}{\sqrt{k}}+\Cr{c6}\epsilon.
\end{equation}
\end{itemize}
Here $\Cr{c4}$, $\Cr{c5}$, $\Cr{c6}$ are the same as in Theorem \ref{thm:main}, and $\Cr{c16}$ depends on $r$, $\Cr{c4}$, and $\Cr{c6}$.
\end{cor}
\begin{proof} We have, from \eqref{eq:l2err} that
\begin{equation}\label{eq:noisy}
\|\hat{x}-x\|_2 \leq \Cr{c4}\left(\frac{m}{\ell}\right)^{-r+1/2}\delta+\Cr{c5}\fr{\sigma_k(x)}{\sqrt{k}}+\Cr{c6} \sqrt{\frac{m}{\ell}}\epsilon
\end{equation}
provided that $m\geq \ell\geq c k\log N$. 

The critical point of the right hand side in \eqref{eq:noisy}, viewed as a function of $\ell$, 
is a minimum, and given by $\ell_c:={m}{\left(\frac{\Cr{c4}(2r-1)\delta}{\Cr{c6} \epsilon}\right)^{-1/r}}$.
If  $c k\log N \leq \ell_c\leq m$ then we get \eqref{eq:noise} by the following argument.  Define $f(\ell):=f_1(\ell)+f_2(\ell)$ where $f_1(\ell):= \Cr{c4}\left(\frac{m}{\ell}\right)^{-r+1/2}\delta$ and $f_2(\ell):=\Cr{c6} \sqrt{\frac{m}{\ell}}\epsilon$. While we may not directly substitute $\ell_c$ into $f(\ell)$ as $\ell$ must be integer valued, we observe  that there is an integer $\ell^*$ between $\ell_c/2$ and $\ell_c$  with 
\begin{align*}f(\ell_c) &\leq f(\ell^*) = f_1(\ell^*)+f_2(\ell^*) \\ &\leq f_1(\ell_c) + f_2(\ell_c/2). \end{align*}
Otherwise if $\ell_c \notin [ c k\log N, m]$, the minimum of $f$ is achieved at one of the two boundary points  $\ell_1^*= {c k \log N}$ (corresponding to \eqref{eq:noise_a}) and $\ell_2^*=m$ (corresponding to \eqref{eq:noise_b}). In particular, $f(m)$ yields \eqref{eq:noise_b}, while $f(\lceil \ell_1^* \rceil) \leq f_1(2\ell_1^*)+ f(\ell_1^*)$ yields \eqref{eq:noise_a}.
\end{proof}

\begin{cor}[Root-exponential accuracy] \label{cor:exp2}
Let $Q_{\sd}^r$ be the $r$th-order coarse $\sd$ quantizer with stability constant $\gamma(r)=\Cr{c3}^r r^r\delta$ as described in Section~\ref{basicsd} and suppose that the hypotheses of Theorem~\ref{thm:main} hold. Let $x\in \R^N$ be $k$-sparse such that $\|\Phi x\|_\infty \le1$ and suppose $q:=Q_{\sd}^{r}(\Phi x)$ with $r=\lfloor\frac{\lambda}{2\Cl{c7}e}\rfloor^{1/2}$ and $\lambda:=\frac{m}{\lceil c k\log N/k \rceil}$. The solution $\hat{x}$ to \eqref{eq:decoder} satisfies
\begin{equation}\label{eq:rtexp}
\|\hat{x}-x\|_2 \leq  \Cl{c8}  e^{-\Cl{c9}\sqrt{\lambda}}\ \delta,
\end{equation}
with probability exceeding $1-e^{-c \rho k\log N/k}$. Here
$\Cr{c7}$ depends only on $\Cr{c3}$, defined after \eqref{eq:SDrec2}, $\Cr{c8}$, $\Cr{c9}$ do
not depend on $m,k,N$, and $c$,$\rho$ are as in Theorem
\ref{thm:main}. \end{cor}
\begin{proof}[Proof of Corollary~\ref{cor:exp2}]
In the case of the coarse $r$th-order $\sd$ quantizer with stability constant $\gamma(r)=\Cr{c3}^r r^r\delta$ and when the signal is sparse and the measurements are noiseless, \eqref{eq:l2err}  reduces to
\begin{align}
\|\hat{x}-x\|_2 &\leq C_4 \left(\frac{m}{\ell}\right)^{-r+1/2} \delta, \notag \\
&=\Cl{c12} (3\pi \Cr{c3})^r r^{2r}\left(\frac{m}{\ell}\right)^{-r+1/2} \delta\label{eq:opt}
\end{align}
where $\Cr{c12}=2^{3/2} \Cr{c10}$. Since this inequality holds for any $\ell\geq c k\log N/k$, it holds for $\ell=\lceil c k\log N/k \rceil$ which yields
\begin{equation}\label{int1}
\|\hat{x}-x\|_2\leq \Cr{c12} (3\pi \Cr{c3})^r r^{2r}  \lambda ^{-r}\sqrt \lambda \ \delta.
\end{equation}
Next, note that for a fixed $\lambda$ (equivalently $m$) the right hand side of \eqref{int1} depends on the order of the $\sd$ scheme, $r$,  and can be optimized yielding
\begin{equation}
\|\hat{x}-x\|_2\leq \Cr{c12} \min_{r\in \N} (3\pi \Cr{c3})^r r^{2r}  \lambda ^{-r}\sqrt \lambda \ \delta.
\end{equation}
As explicitly shown in \cite{G-exp}, the critical point of the function $f(t)=(3\pi \Cr{c3})^t t^{2t}  \lambda ^{-t}\sqrt \lambda$ is given by 
$$t^*=e^{-1}(3\pi \Cr{c3})^{-1/2}\lambda^{1/2}.
$$
Accordingly, setting $r=\lfloor t^* \rfloor$ and simplifying, we obtain
\begin{align}
\|\hat{x}-x\|_2 &\leq \Cr{c12} e^2 e^{-2(3\pi \Cr{c3})^{-1/2} \sqrt{\lambda}} \sqrt{\lambda} \ \delta \notag\\
& \leq \frac{1}{2}\Cr{c12} e^2 (3\pi \Cr{c3})^{1/2} e^{-(3\pi \Cr{c3})^{-1/2} \sqrt{\lambda}} \delta   \notag.
\end{align}
 Setting $\Cr{c8}=\frac{1}{2}\Cr{c12} e^2 (3\pi \Cr{c3})^{1/2}$ and $\Cr{c9}=(3\pi \Cr{c3})^{-1/2}$ we arrive at \eqref{eq:rtexp}.

%
\end{proof}

Next, we consider the case when the signal is compressible and the measurements are noise-free. In this special case, \eqref {eq:l2err} reduces to
\begin{equation}
\|\hat{x}-x\|_2 \leq \Cr{c4} \left(\frac{m}{\ell}\right)^{-r+1/2}\delta+\Cr{c5} \frac{\sigma_k(x)}{\sqrt k}, 
\end{equation}
where $\ell$ and $k$ can be chosen freely as long as $m\ge \ell \ge ck\log N/k$. Below, we optimize this upper bound for compressible signals, modeled as elements in weak $\ell_p$ space.  Specifically, we say that $x$ is in the unit ball of the weak $\ell_p$ space $w\ell_p^N$ if  $|x_{(j)}| \leq
  j^{-1/p}$ where $x_{(j)}$ is the $j$th largest-in-magnitude entry of
  $x$.

%
%

\begin{cor}[Compressible signals]\label{cor:compr} 
  Let $\Phi$ be an $m\times N$
 sub-Gaussian matrix with mean zero and unit variance, and let $k$ and
  $\ell$ be in $\{1,\dots,m\}$. Denote by $Q_{\sd}^r$ a stable
  $r$th-order scheme with alphabet $\A_K^\delta$.  Fix $p\in(0,1)$ and set $\alpha:=\frac{r-1/2}{r+1/p-1}$.  
  There exist positive constants $c$ and $\rho$ (as defined in Theorem~\ref{thm:main}) such that whenever $m\ge \Cl{c13}c\log N$, the following holds with probability greater than $1-\exp(-c\rho m^\alpha\log^{1-\alpha} N)$ on the draw of $\Phi:$
\bigskip

For vectors $x\in \R^N$ in the unit ball of $w\ell_p^N$, let $q:=Q_{\sd}^r(\Phi x)$.  Then the solution $\hat{x}$ to \eqref{eq:decoder} with $\epsilon=0$ satisfies 
\begin{equation}\label{eq:l2err_compr}
\|\hat{x}-x\|_2\leq \Cl{c14} \delta^\frac{1/p-1/2}{r+1/p-1} \left(\frac{m}{(c+1) \log N}\right)^{-\frac{(1/p-1/2)(r-1/2)}{r+1/p-1}}.
\end{equation}
Here $\Cr{c13}$ and $\Cr{c14}$ depend on $r,t,\delta$ and are given explicitly below.
\end{cor}
\begin{proof} It will be notationally less cumbersome to work with $t=1/p$. For any $x$ in the unit weak $\ell_{1/t}$ ball, its tail (measured in the $\ell_1$ norm) decays as
\[
\sigma_k(x) \leq \int\limits_{k+1}^N s^{-t}  \mathrm{d}s \leq \frac{k^{1-t}}{t-1}.
\]
Inserting this bound and $\epsilon=0$ into \eqref{eq:l2err}, we get
\[
\|\hat{x}-x\|_2 \leq \Cr{c4} \left(\fr{m}{\ell}\right)^{-r+1/2}\delta+\frac{\Cr{c5}}{t-1}  k^{-t+1/2} , 
\]
Note that for any fixed $k$, we can make the upper bound above smallest by choosing the smallest $\ell$ for which \eqref{eq:l2err} holds, i.e., $\ell=\lceil c k \log N/k \rceil$. With this choice, we now have 
\begin{align} \label{eq:compr}
\|\hat{x}-x\|_2 &\leq \Cr{c4}\left(\fr{m}{\lceil c \ k \log N/k\rceil}\right)^{-r+1/2}\delta+\frac{\Cr{c5}}{t-1}  k^{-t+1/2} \notag \\& \leq \Cr{c4}\left(\fr{m}{c' \ k \log N}\right)^{-r+1/2}\delta+\frac{\Cr{c5}}{t-1}  k^{-t+1/2}
\end{align}
Above, replacing $c$ by $c':=c+1$ allows us to remove the ceiling function, and we relaxed $\log N/k$ to $\log N$ to make optimizing over $k$ easier in the next step. 
To simplify the ensuing argument define $f_1(k):=\Cl{c4p}\left(\fr{m}{c'\ k \log N}\right)^{-r+1/2}\delta$ and $f_2(k):=\frac{\Cl{c5p}}{t-1}  k^{-t+1/2}$ where $\Cr{c4p}:=\max\limits_{k \in [1, m/(c\log N)]} C_4(k)$ and $\Cr{c5p}:=\max\limits_{k \in [1, m/(c\log N)]} C_5(k)$.
The critical point of $f(k):=f_1(k)+f_2(k)$, the right hand side in \eqref{eq:compr}, is
\[
k_c=\left(\frac{\Cl{c15}}{\delta} \left(\frac{m}{c'\ \log N}\right)^{r-1/2} \right)^{\frac{1}{r+t-1}}
\]
where $\Cr{c15}=\frac{\Cr{c5p} (t-1/2)}{\Cr{c4p} (t-1)(r-1/2)}$. In particular, this critical point is readily seen to be a minimum. However, for  \eqref{eq:compr} to hold, $k$ must be an integer satisfying $1\leq k\leq \frac{m}{c \log N}$ so we can not directly substitute $k_c$ in \eqref{eq:compr}. Nevertheless, we note that $k_c$ is a minimum, $f_1(k)$ is increasing with $k$, and $f_2(k)$ is decreasing with $k$. Moreover, if $k_c\geq 1$ there exists an integer $k^*$ between $k_c$ and $k_c/2$ with \begin{align*}f(k_c) &\leq f(k^*) = f_1(k^*)+f_2(k^*) \\ &\leq f_1(k_c) + f_2(k_c/2).\end{align*} Above, the first inequality is by the optimality of $k_c$ and the second by the monotonicity of $f_1$ and $f_2$. 
Thus, if $1\leq k_c\leq \frac{m}{c' \log N}$, or equivalently if $$m\geq \max\left\{c' \log N \left(\Cr{c15}/\delta\right)^{\frac{1}{t-1/2}},  c' \log N \left(\delta/\Cr{c15}\right)^{\frac{1}{r-1/2}}\right\}$$ this yields 
\[
\|\hat{x}-x\|_2\leq \Cr{c14} \delta^{\frac{t-1/2}{r+t-1}} \left(\frac{m}{c' \log N}\right)^{-\frac{(t-1/2)(r-1/2)}{t+r-1}},
\]
where $\Cr{c14}= \Cr{c4}\Cr{c15}^\fr{r-1/2}{r+t-1} + \fr{2^{t-1/2}}{t-1}\Cr{c5}\Cr{c15}^\fr{-t+1/2}{r+t-1}$.
By Theorem \ref{thm:main}, the above bound holds with probability exceeding $1-e^{-\rho c k_c\log N}=1-\exp\left(-\rho c \ m^{\fr{r-1/2}{r+t-1}} \ \log^{\fr{t-1/2}{r+t-1}} N\right)$. Setting $\Cr{c13}=\max\left\{\left(\Cr{c15}/\delta\right)^{\frac{1}{t-1/2}},  \left(\delta/\Cr{c15}\right)^{\frac{1}{r-1/2}}\right\}$, $\alpha=\frac{r-1/2}{r+t-1}$, we obtain \eqref{eq:l2err_compr}.

\end{proof}

\begin{rem} The results above place no restrictions on how large $m$ can grow, compared to $N$. In particular,  our results apply to both the ``undersampling" $(m\leq N)$ and ``oversampling" $m\geq N$ cases. However,  they are not necessarily optimal for the $m \geq N$ case. In this setting,  the Sobolev dual decoder proposed in \cite{blum:sdf}, when applied to $\sd$ quantized sub-Gaussian measurements of compressible signals yields an error decay of order $\left(\frac{m}{N}\right)^{\alpha(r-1/2)}$, where $\alpha\in (0,1)$ controls  the probability \cite{GLPSY13, KSY13}. This error is better than predicted by Corollary \ref{cor:compr}. Below, we show that our one-stage decoder also yields this quantization error decay in the oversampled case, while maintaining robustness to noise.

\begin{cor} Assume all the variables are as defined in Theorem \ref{thm:main} with $m\geq \ell \geq dN$, where $d\geq 18$ is a constant.  Then there exists a constant $c$ such that decoding via \eqref{eq:decoder} yields the uniform error bound 
\[
\|\hat{x}-x\|_2 \leq \Cr{c4} \left(\frac{m}{\ell}\right)^{-r+1/2}\delta+\Cr{c6}\sqrt{\frac{m}{\ell}}\epsilon , 
\]
with probability $1-e^{-c N}$. 
\end{cor}
\begin{proof} One can repeat the proof of Theorem  \ref{thm:main}, with minor modifications, to obtain the result. Specifically, applying Proposition \ref{pro:singular_sub} with $t=\sqrt{\frac{N}{\ell}}$ we see that with probability exceeding $1-e^{-\Cl{c19} N}$, the matrix $ \widetilde \Phi$ (as defined in the proof of Theorem \ref{thm:main}) satisfies the restricted isometry property with constant $\delta =\frac{2N}{m}\leq \frac{2}{d} \leq 1/9$, and sparsity $k=N$. Everything else remains unchanged and yields the desired result by setting $\sigma_k(x)=0$.
 \end{proof}

\section{Numerical Experiments}
In this section, we present numerical experiments to illustrate our theoretical results. In all the simulations below, we use random Gaussian matrices as sensing matrices.  
\subsection*{Sparse signals}
In the first experiment, we compare the empirical decay rate of the
reconstruction error due to solving \eqref{eq:decoder}, with those
predicted in Theorem \ref{thm:main}. We generate a Gaussian matrix
$\Phi_0$ of size $1000 \times 512$. We then generate $m\times 512$
mesurement matrices for various values of $m$ between 100 and 1000 by
selecting the first $m$ rows of $\Phi_0$. We also generate a set of
100 $k$-sparse vectors, $k=10$, where the support is chosen uniformly
at random and the magnitudes are such that the restriction of $x$ to
its support is uniformly distributed in the unit-ball of $\R^k$. We
set the quantization step size $\delta= 0.01$ and use $r$th-order
$\sd$ quantization, $r=1,2$, with the corresponding greedy
quantization rule to quantize the $m$ compressed sensing measurements
of each of the 100 sparse vectors. Figure \ref{fig:sparse} shows the
worst-case reconstruction error over the 100 experiments described
above as a function of the number of measurements $m$. Note that the
behavior of the error is as expected.  
\begin{figure}[h!]
\begin{centering}
\includegraphics[scale=0.4]{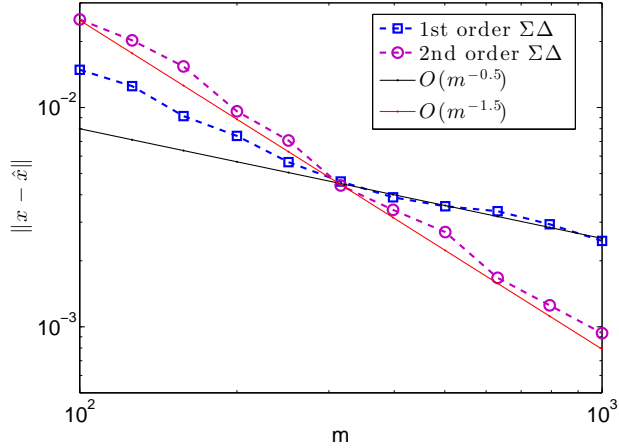}
\caption{Worst case reconstruction error of \eqref{eq:decoder} on a log-log plot, with parameters: \mbox{$N=512$}; $m=\lfloor10^{2+0.1l}\rfloor$ for $l=0,1,..,10$; $k=10$; $\delta=0.01$. 
}
\label{fig:sparse}
\end{centering}
\end{figure}

\subsection*{One-bit quantization} 
Next, we repeat this experiment (albeit with $N=256$, $k=5$)
for one-bit $\sd$ quantization of compressed sensing measurements, using
\eqref{eq:SDrec2} with the one-bit alphabet $\{\pm1\}$ as the
quantizer and \eqref{eq:decoder2} as the decoder. To ensure stability
of the quantizer for $r=1,2$, we consider sparse signals lying inside
the $\ell_2$-ball with radius 0.15 (alternatively, we could scale the
alphabet and work with vectors in the unit ball).  Figure
\ref{fig:onebit} plots the corresponding result with each point
representing the worst case error over 100 experiments. The predicted
rates, $m^{-0.5}$ and $m^{-1.5}$, appear in the Figure
\ref{fig:onebit} as $m$ gets large. The initial plateaus of both
curves correspond to the range of $m$ for which the zero-vector is
feasible and thus optimal. As $m$ increases to a point where zero is no longer
feasible to \eqref{eq:decoder2} the predicted decay rate appears.
\begin{figure}[h!]
\begin{centering}
\includegraphics[scale=0.4]{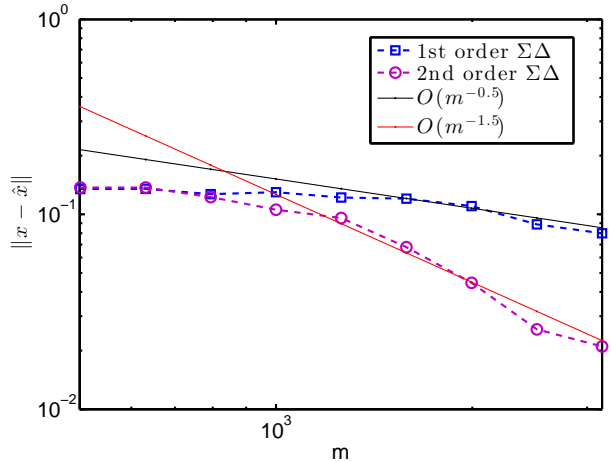}
\caption{Log-log plot for the worst case reconstruction error of one-bit quantization with parameters $N=256, m=\lfloor 10^{2+0.1l} \rfloor, l=7,...,15$, and $k=5$.
\label{fig:onebit}
}
\label{fig:}
\end{centering}
\end{figure}

 \subsection*{Compressible signals}
In the next experiment, we consider compressible signals. Our measurement setup is
identical to that of the first experiment. Specifically, we use
restrictions of a $1000\times 512$ Gaussian matrix
$\Phi_0$ to generate $m\times 512$ measurement matrices for
various values of $m$ between 100 and 1000. We generate 100 compressible signals in the unit $w\ell_p^N$ ball with $p=0.5$, so $|x_{(i)}| \leq \frac{1}{i^{2}}$ where $x_{(i)}$ denotes the $i^{\text{th}}$ largest component of $x$ in magnitude. 
Specifically, each compressible signal is generated by randomly permuting the indices of a vector whose $i$'th element is drawn from the uniform distribution on the interval $[-1/i^2,1/i^2]$.    Figure \ref{fig:compressible} compares the decay rate (as a function of $m$) of the worst case error over the 100 signals with those predicted in Corollary \ref{cor:compr}, i.e., 
 $m^{-0.375}$ and $m^{-0.75}$ for $r=1,2$, respectively.  
 
 \begin{figure}[h!]
\begin{centering}
\includegraphics[scale=0.4]{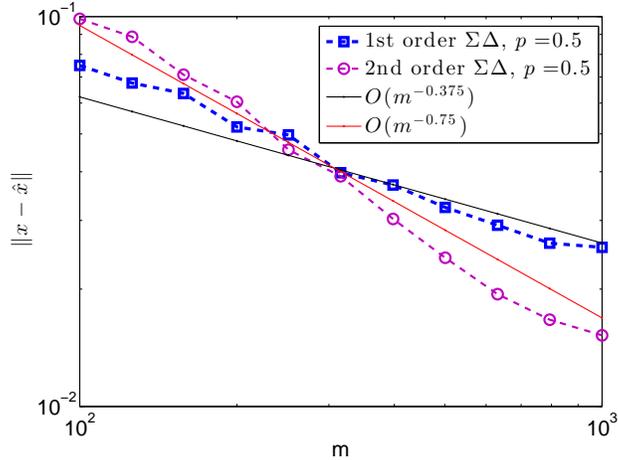}
\caption{Log-log plot for the worst case reconstruction error of compressible signals in the unit $w \ell_p^N$ ball with $p=1/2$, $N=512, m=\lfloor 10^{2+0.1l} \rfloor, l=0,1,...,10$, and $k=10$.
}
\label{fig:compressible}
\end{centering}
\end{figure}
 
\subsection*{Robustness to noise}
This experiments illustrates our decoder's robustness to noise. We use the same set of parameters as in the first experiment,  except that now the measurements are damped by i.i.d., uniform noise $\eta$ with $\|\eta\|_{\infty} \leq 2\cdot 10^{-3}$. Figure \ref{fig:noisy}  provides evidence that as $m$ grows, the limit of the reconstruction error approaches a fixed value as predicted in \eqref{eq:noise} and \eqref{eq:noise_b} of Corollary \ref{cor:noise}. 
\begin{figure}[h!t!]
\begin{centering}
\includegraphics[scale=0.4]{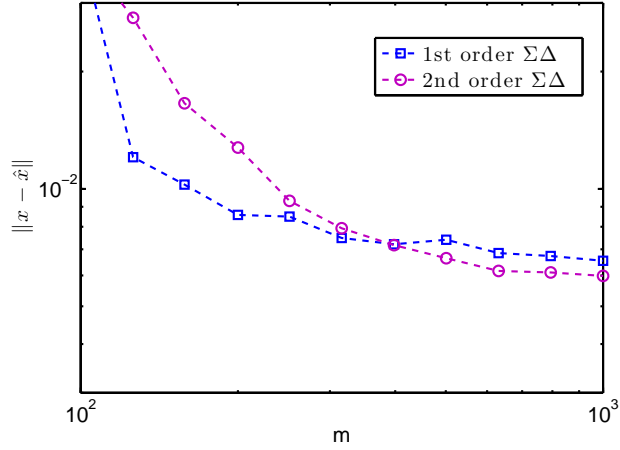}
\caption{Worst case reconstruction error with noisy measurements setting $\epsilon=2\cdot 10^{-3}$, $N=512, m=\lfloor 10^{2+0.1l} \rfloor, l=0,1,...,10$, and $k=10$.}
\label{fig:noisy}
\end{centering}
\end{figure}

\subsection*{Root-exponential error-decay}
Our final experiment shows the root exponential decay predicted in Corolloary \ref{cor:exp2}. Here, for the compressed sensing measurements of each generated signal, we apply $\sd$ quantization with different orders (from one to five) and select the one with the smallest reconstruction error. For this experiment, the setup is again similar to our first experiment, with $m$ ranging from $100$ to $398$ while  $N=512$, $k= 5$, and $\delta=0.01$. Since the sparsity level  is fixed, the oversampling factor $\lambda$ is propositional to $m$ and we expect the reconstruction error to decay root exponentially in $m$,
$$\|\hat{x}_{opt}-x\| \leq e^{-c \sqrt m},$$ 
or equivalently
$$\log(-\log(\|\hat{x}_{opt}-x\|)) \geq \frac{1}{2} \log m+\log c.$$ 
Thus, in Figure \ref{fig:root_exp}, we plot  (as a function of $m$), the $\log(-\log(\cdot))$ of the mean reconstruction error over 50 experiments as well as $\frac{1}{2} \log m+\log c$.

\begin{figure}[h!]
\begin{centering}
\includegraphics[scale=0.4]{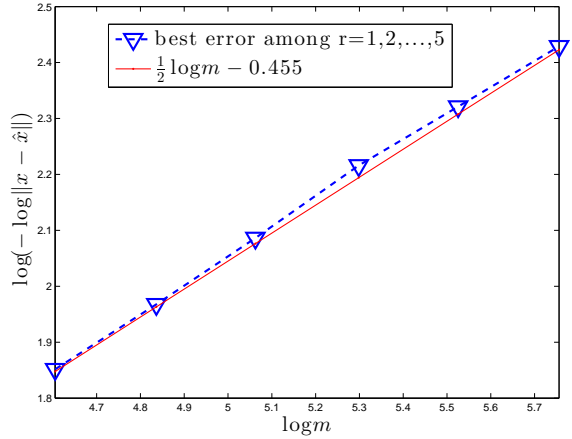}
\caption{Root exponential decay, with $N=512$; $k=5$; $N=512$, and $m=\lfloor 10^{2+0.1l} \rfloor, l=0,1,...,10$ .
}
\label{fig:root_exp}
\end{centering}
\end{figure}
\section{Acknowledgements} 
R.~Wang was funded in part by an NSERC
Collaborative Research and Development Grant DNOISE II (22R07504). {\"O}.~Y{\i}lmaz was funded in part by a Natural Sciences and
Engineering Research Council of Canada (NSERC) Discovery Grant
(22R82411), an NSERC Accelerator Award (22R68054) and an NSERC
Collaborative Research and Development Grant DNOISE II (22R07504).

\bibliographystyle{plain}
\bibliography{DCCbib,nsfrefs,refs}
\end{document}